\documentclass[sigconf]{acmart}

\copyrightyear{2021}
\acmYear{2021}
\setcopyright{rightsretained}
\acmConference[ISSAC '21]{International Symposium on Symbolic and Algebraic Computation}{July 18--22, 2021}{Saint Petersburg, Russia}
\acmBooktitle{International Symposium on Symbolic and Algebraic Computation (ISSAC '21), July 18--22, 2021, Saint Petersburg, Russia}
\acmPrice{15.00}
\acmDOI{XX.XXX/XXXXXX.XXXXXX}
\acmISBN{XXXXXXXXXXXXXXXXXXX}

\usepackage[utf8]{inputenc}

\usepackage{hyperref}

\usepackage{amsmath}
\usepackage{amsthm}
\usepackage{mathrsfs}
\usepackage{stmaryrd}
\usepackage{enumerate}
\usepackage[algoruled,vlined,english,linesnumbered,algonl]{algorithm2e}

\SetCommentSty{mycommfont}
\usepackage{comment}
\usepackage{multirow}
\usepackage{xspace}
\usepackage{verbatim}

\usepackage{mathtools}
\mathtoolsset{showonlyrefs,showmanualtags}

\usepackage[inline]{enumitem}

\usepackage{array}
\newcommand{\PreserveBackslash}[1]{\let\temp=\\#1\let\\=\temp}
\newcolumntype{C}[1]{>{\PreserveBackslash\centering}p{#1}}

\usepackage{tikz}
\usetikzlibrary{calc, matrix, backgrounds,shapes.geometric}

\usepackage{mathtools}
\mathtoolsset{showonlyrefs,showmanualtags}

\usepackage[inline]{enumitem}

\definecolor{input}{HTML}{303060}
\definecolor{output}{HTML}{804000}
\definecolor{string}{HTML}{A02020}
\definecolor{parent}{HTML}{A020A0}
\definecolor{function}{HTML}{205080} 
\definecolor{constructor}{HTML}{205080}
\definecolor{method}{HTML}{205080}
\definecolor{keyword}{HTML}{008000}
\definecolor{error}{HTML}{B01010}
\definecolor{comment}{HTML}{60A060}

\SetArgSty{textup}
\SetKwInOut{Input}{Input}\SetKwInOut{Output}{Output}
\renewcommand{\epsilon}{\varepsilon}

\newcommand{\noopsort}[1]{}

\DeclareMathOperator{\val}{val}

\newcommand{\N}{\mathbb N}

\newcommand{\NN}{\mathbb N}
\newcommand{\ZZ}{\mathbb Z}

\newcommand{\Q}{\mathbb Q}
\newcommand{\QQ}{\mathbb Q}
\newcommand{\R}{\mathbb R}
\newcommand{\Qp}{\Q_p}

\newcommand{\FF}{\mathbb{F}}

\newcommand{\Kz}{K^\circ}
\newcommand{\KzX}[1][]{K\{ \X \ifnonempty{#1}{; #1}{} \}^\circ}
\newcommand{\KX}[1][]{K \{ \X \ifnonempty{#1}{; #1}{} \}}

\newcommand{\Lz}{L^\circ}

\renewcommand{\mod}{\;\mathrm{mod}\;}

\newcommand{\X}{\mathbf{X}}
\newcommand{\Y}{\mathbf{Y}}
\renewcommand{\i}{\mathbf{i}}

\renewcommand{\r}{\mathbf{r}}
\renewcommand{\u}{\mathbf{u}}

\newcommand{\smallsp}{\textup{Small}}
\newcommand{\bigsp}{\textup{Big}}
\newcommand{\Sat}{\textup{Sat}}
\renewcommand{\prec}{\textup{prec}}

\newcommand{\softO}{O\tilde{~}}

\newcommand{\floor}[1]{\lfloor #1 \rfloor}

\newcommand{\ifnonempty}[3]{%
  \def\tempa{}%
  \def\tempb{#1}%
  \ifx\tempa\tempb 
  #3 
  \else            
  #2
  \fi}

\makeatletter
\newcommand{\removelatexerror}{\let\@latex@error\@gobble}
\makeatother

\DeclareMathOperator{\LT}{LT}

\DeclareMathOperator{\LM}{LM}
\DeclareMathOperator{\NF}{NF}

\definecolor{purple}{rgb}{0.6,0,0.6}

\definecolor{answer}{rgb}{0,0.5,0.2}

\newtheorem{theo}{Theorem}[section]
\newtheorem{lem}[theo]{Lemma}
\newtheorem{prop}[theo]{Proposition}

\theoremstyle{definition}
\newtheorem{rem}[theo]{Remark}

\copyrightyear{2021}
\acmYear{2021}
\setcopyright{rightsretained}
\acmConference[Conference'21]{Conference'21}{July 2021}{Washington, DC, USA}
\acmBooktitle{}
\acmPrice{15.00}
\acmDOI{}
\acmISBN{}

\begin{document}

\title{On FGLM Algorithms with Tate Algebras}

\author{Xavier Caruso}
\affiliation{Université de Bordeaux,
  \institution{CNRS, INRIA}
  \city{Bordeaux}
  \country{France}}
\email{xavier.caruso@normalesup.org}

\author{
  Tristan Vaccon}
\affiliation{Universit\'e de Limoges;
  \institution{CNRS, XLIM UMR 7252}
  \city{Limoges}
  \country{France}  
  \postcode{87060}  
}
\email{tristan.vaccon@unilim.fr}

\author{
  Thibaut Verron}
\affiliation{Johannes Kepler University, 
  \institution{Institute for Algebra}
  \city{Linz}
  \country{Austria}  
}
\email{thibaut.verron@jku.at}

\thanks{This work was supported by the ANR project CLap--CLap
(ANR-18-CE40-0026-01).
T.~Verron was supported by the Austrian FWF grant P31571-N32.}

\begin{abstract}
Tate introduced in~\cite{Tate} the notion of Tate algebras to serve, in
the context of analytic geometry over the $p$-adics, as a counterpart
of polynomial algebras in classical algebraic geometry. In~\cite{CVV,
CVV2} the formalism of Gröbner bases over Tate algebras has been
introduced and advanced signature-based algorithms have been proposed.
In the present article, we extend the FGLM algorithm of
\cite{Faugere:1993} to Tate algebras. Beyond allowing for fast change
of ordering, this strategy has two other important benefits.
First, it provides an efficient algorithm for changing the radii of
convergence which, in particular, makes effective the bridge between
the polynomial setting and the Tate setting and may help in speeding
up the computation of Gröbner basis over Tate algebras.
Second, it gives the foundations for designing a fast algorithm for
interreduction, which could serve as basic primitive in our previous
algorithms and accelerate them significantly.
\end{abstract}



\vspace{-1.5mm}
\terms{Algorithms, Theory}


\maketitle

\section{Introduction}

Lying at the intersection of geometry and number theory, one finds $p$-adic
geometry. A paramount part of this theory is the study of $p$-adic analytic 
varieties, 
first defined by Tate in~\cite{Tate} (see also \cite{FP04}).
They have played a key role in many developments of number
theory (\emph{e.g.} 
$p$-adic cohomologies~\cite{LS07}, $p$-adic modular forms~\cite{Go88}).
The main algebraic objects upon which Tate's geometry is built are
Tate algebras and their ideals, formed of convergent multivariate power series over 
a complete discrete valuation field~$K$ (\emph{e.g.} $K = \Qp$).

In earlier papers~\cite{CVV, CVV2}, the authors showed that it is possible to define and compute Gröbner bases of Tate ideals with coefficients in $\ZZ_{p}$ or $\QQ_{p}$, and that the definitions are compatible with the usual theory on polynomials over the residue field $\FF_{p}$ or over the coefficient ring.
A major limitation of the algorithms is the increasing cost of reductions as the precision grows.
Our previous paper~\cite{CVV2} addresses the case of expensive reductions to zero, through the use of signature algorithms, but computing the result of non-trivial reductions remains expensive.
Another question left open was whether it is possible to exploit overconvergence properties, namely the knowledge that the series we are working with satisfy a stronger convergence condition.

In the present paper, we adapt the classical FGLM algorithm to the case 
of Tate series, and we show that it gives answers to both 
questions, in the case of zero-dimensional ideals. Precisely, we 
prove the following theorem.

\begin{theo}
\label{theo:intro}
Let $\KX[\r]$\footnote{Here $K$ 
denotes the base field and $\r$ encodes the radii of convergence of 
our series; we refer to \S \ref{subsec:Tate_alg} for the precise 
definitions.} and $\KX[\u]$ be two Tate algebras with
$\KX[\r] \subset \KX[\u]$.

There exists an algorithm that takes as input a reduced Gröbner basis 
$G$ of a $0$-dimensional ideal $I$ of $\KX[\r]$ with respect to a given 
monomial ordering and output a Gröbner basis of the ideal 
$I{\cdot}\KX[\u]$ of $\KX[\u]$ for another given monomial ordering.

Moreover, if $n$ denotes the number of variables, if $\delta$ is the 
dimension of the quotient $\KX[\r]/I$ and if $\prec$ is the precision 
at which the result 
is output, the complexity of this algorithm is:
\begin{itemize}
\item
$\softO(n \delta^3 \prec)$  operations in the base field $K$ 
for a general $K$,
\item $\softO(n 
\delta^3 \prec \cdot \log p)$ bit operations when $K = \Qp$.
\end{itemize}
\end{theo}

We underline that, although the classical FGLM 
algorithm only concerns change of ordering, our version also permits 
to change the radii of convergence of the underlying Tate algebra
(namely the parameters $\r$ and $\u$), 
and then provides efficient tools for dealing with the aforementioned 
overconvergence situation. In the extreme case where $\r$ is infinite,
it makes effective the bridge between polynomials and Tate series,
that is between classical algebraic geometry and rigid geometry.
On a different note, being able to perform such a change of ordering 
opens up algorithmic strategies for overconvergent series, by giving 
freedom in the choice of the convergence radii.

An additional important outcome of our algorithm is that it can be
slightly modified in order to accept certain nonreduced Gröbner 
bases as input. Hence, in many cases, calling it with the 
same radii of convergence and the same ordering as input and output, 
already performs a nontrivial operation: the interreduction of the 
input Gröbner basis. Moreover, it has a controlled 
complexity and performs actually very well in practice (contrarily 
to the naive reduction algorithm).
Since the intermediate interreduction of Gröbner bases is often
the bottleneck in Buchberger and signature algorithms in the Tate
setting, using our FGLM algorithm (or an adaptation of it) at
this step could lead to a significant speed-up.

\subsubsection*{Strategy and ingredients}

In the classical setting, the key step of the FGLM algorithm is to convert back and forth between Gröbner bases and the so-called multiplication matrices, which are defined as the multiplication maps by the variables in the quotient space. Performing the change of ordering on those multiplication matrices then reduces to basic linear algebra.
Still in the classical case, thanks to the structure of normal forms, it can be shown that all steps can be done in sub-cubic time in the number of solutions.

In the Tate setting, Gröbner bases are defined using a term ordering, taking into account both a monomial ordering in the usual sense and a weight taking into account the degree of the monomials, the valuation of the coefficient and the convergence radius of the series in the algebra.
It is the reason why we will eventually be able to change all these 
parameters at the same time. However, this feature also implies new
difficulties.

Firstly, in the construction of the multiplication matrices, the structure of the normal forms does not allow us to read the values in one pass.
Instead, we prove that an iterative process converges to the correct value of the matrices, and we show how this process can be done in different ways, including, for some particular base fields,
the option of using relaxed arithmetic~\cite{vdH:1997,BvdHL:2011}, which eventually leads to a significant improvement of the efficiency.

Secondly, if the change of ordering incurs a change of convergence radii, the size of the quotient algebra might change.
We show that it is possible to recover multiplication matrices over the correct quotient by separating eigenspaces depending on the valuation of the eigenvalues.
The reconstruction of the final Gröbner basis is finally achieved using the classical strategy in the residue field, and then lifting the basis.

\subsubsection*{Organization of the article}

In Section~\ref{sec:setting}, we introduce the notations and
discuss some primitives of linear algebras over nonarchimedian fields
which will be used repeatedly later on.
The computation of multiplication matrices is addressed in
Section~\ref{sec:mult_mat_for_Tate_alg}.
In Section~\ref{sec:FGLM_for_Tate_alg}, we consider the question 
of changing radii of convergence and design our final algorithm.


\section{Setting and preliminaries}
\label{sec:setting}

Throughout this article, we consider a field $K$ equipped with a 
discrete valuation $\val$ for which it is complete.
We denote its ring of integers by $\Kz$ and fix a uniformizer
$\pi$ of $K$. The quotient $\Kz/\pi$ is called the residue field
of $K$ and will be denoted by $\bar K$ in what follows.
Classical examples of such fields are $K = \mathbb{Q}_p$ (equipped 
with the $p$-adic valuation) and $k(\!(T)\!)$ (equipped with the $T$-adic 
valuation) for any base field $k$.

The complexity statements are given with the usual asymptotic notations $O(f)$ and $\softO(f) = O(f \log(f)^{n})$ for some $n$.

We will consider two different models of complexity: arithmetic 
complexity, counting operations in $K$ or $\Kz$, and base complexity, 
taking into account the precision.
In the case of equal characteristic (\emph{i.e.} $\text{char } K = 
\text{char } \bar K$), such as $k(\!(T)\!)$, the base complexity 
counts operations in the residue field, and the correspondence 
between both models satistfies:
\begin{equation}
  \label{eq:4}
  \text{(Arithmetic complexity)} = 
  \softO\big(\text{(Base complexity)} \cdot \prec\big).
\end{equation}
where $\text{prec}$ stands for the working precision.
On the contrary, in the case of mixed characteristic, such as 
$\QQ_{p}$, the base 
complexity counts bit operations. When the residue field is
finite, the correspondence between both models satisfies:
\begin{equation}
  \label{eq:5}
  \text{(Arithmetic complexity)} =
  \softO\big(\text{(Base complexity)} \cdot \prec \cdot \log |\bar K|\big).
\end{equation}

\subsection{Tate algebras and ideals}
\label{subsec:Tate_alg}

In order to fix notations, we briefly recall the definition of Tate 
algebras and the theory of Gröbner bases over them.
Let $\r = (r_{1},\dots,r_{n}) \in \QQ^{n}$.
The \emph{Tate algebra} $\KX[\r]$ is defined by:
  \begin{equation}
    \label{eq:1}
    \KX[\r] := \left\{ \sum_{\mathbf{i} \in \NN^{n}} a_{\mathbf{i}}\X^{\mathbf{i}} 
     \text{ s.t. }
     a_{\mathbf{i}}\in K \text{ and } 
     \val(a_\i) - \r{\cdot}\i \xrightarrow[|\i| \rightarrow +\infty]{} +\infty
    \right\}
  \end{equation}
The tuple $\mathbf{r}$ is called the convergence log-radii of the 
Tate algebra.
We define the Gauss valuation of a term 
$a_{\mathbf{i}}\X^{\mathbf{i}}$ as 
$\val(a_{\mathbf{i}}\X^{\mathbf{i}}) = \val(a_\i) - \r{\cdot}\i$, and 
the Gauss valuation of $\sum a_{\mathbf{i}}\mathbf{X}^{\mathbf{i}} 
\in \KX[\r]$ as the minimum of the Gauss valuations of its terms.
The integral Tate algebra ring $\KzX[\r]$ is the subring of $\KX[\r]$ 
consisting of elements with nonnegative valuation.
In what follows, when $\r = (0, \ldots, 0)$, we will simply write 
$\KX$ instead of $\KX[(0,\ldots,0)]$.

We fix a classical \emph{monomial order} $\leq_m$ on the set of
monomials $\X^{\mathbf{i}}$. 
Given two terms $a \X^\i$ and $b \X^{\mathbf{j}}$  (with $a,b \in 
K^\times$), we write $a \X^\i < b \X^{\mathbf{j}}$ if
$\val(a\X^\i) > \val(b\X^{\mathbf{j}})$, or
$\val(a\X^\i) = \val(b\X^{\mathbf{j}})$ 
and $\X^{\mathbf{i}} <_m \X^{\mathbf{j}}$.
The leading term of a Tate series $\sum 
a_{\mathbf{i}}\mathbf{X}^{\mathbf{i}} \in \KX[\r]$ is, by
definition, its maximal term.

A Gröbner basis of an ideal $I$ of $\KX[\r]$ is,
by definition, a family $(g_1, \ldots, g_s)$ of elements of $I$
with the property that for all $f \in I$, there exists an index
$i \in \{1, \ldots, s\}$ such that $\LT(g_i)$ divides $\LT(f)$.
A Gröbner basis $(g_1, \ldots, g_s)$ is \emph{reduced} if given a term $t$
of $g_i$ which is not the leading term, $t$ is not divisible
by any $\LT(g_j)$. The following theorem is proved in~\cite{CVV}.

\begin{theo}
\label{theo:GB}
\begin{enumerate}
\item
Any ideal of $\KX[\r]$ admits a Gröbner basis.
\item
If $\r = (0, \ldots, 0)$ and $I$ is an ideal of $\KX$, a family
$G = (g_1, \ldots, g_s)$ consisting of elements of $\KX$ with Gauss 
valuation~$0$ is a Gröbner basis of $I$ if and only if its reduction
modulo~$\pi$ is a classical Gröbner basis of the quotient ideal
$(I \cap \KzX[\r]) / \pi(I \cap \KzX[\r])$ of $\bar K[\X]$ for $\leq_m$.
\end{enumerate}
\end{theo}

In the present article, we will be particularly interested in
$0$-dimensional ideals. By definition, $I$ is such an ideal if
the quotient $\KX[\r]/I$ is a finite dimensional $K$-vector
space.
If $I$ is a $0$-dimen\-sional ideal, the set:
$$B = \big\{\, \X^{\mathbf i} 
\text{ with } i \in \N^n
\text{ and } \X^{\mathbf i} \not\in \LT(I) \, \big\}$$
is finite and forms a $K$-basis of $\KX[\r]/I$. 
It is called the \emph{staircase} of $G$. Moreover, if we are
given a Gröbner basis $(g_1, \ldots, g_s)$ of $I$, the staircase $B$
consists of all monomials $\X^{\mathbf i}$ which are not divisible by
any $\LT(g_j)$ for $j$ varying in $\{1, \ldots, s\}$.
This observation implies in particular that any \emph{reduced} Gröbner basis
of a $0$-dimensional ideal $I$ consists only of polynomials.

\subsection{Linear algebra}
\label{ssec:linalg}

It is an understatement to say that the FGLM strategy relies 
heavily on linear algebra. In the Tate setting, this assertion
is even more true and new basic operations in linear algebra,
which are specific to non-archimedean base fields, will be
needed.
The aim of this subsection is to review briefly these operations.

\subsubsection*{Slope decomposition.}

Let $V$ be a finite $K$-dimensional vector space and let
$f : V \to V$ be a $K$-linear mapping.
Let $\chi_f$ be the characteristic polynomial of $f$.
Given an auxiliary real number $s$, one can factor $\chi_f$ as
a product
$\chi_f = \chi_{\bigsp,s, f} \times \chi_{\smallsp,s, f}$
where $\chi_{\bigsp,s, f}$ (resp. $\chi_{\smallsp,s,f}$) is the 
factor corresponding to all roots (in an algebraic closure) of 
valuation $<s$ (resp. valuation $\geq s$).
Moreover, both $\chi_{\bigsp,s, M}$ and $\chi_{\smallsp,s, M}$
have coefficients in $K$.
Letting $\bigsp_s(f)$ denote the kernel of $\chi_{\bigsp,s, f}(f)$
and $\smallsp_s(f)$ denote that of $\chi_{\smallsp,s, f}(f)$, the
above factorization corresponds to a decomposition of $V$ as a 
direct sum
$V = \bigsp_s(f) \oplus \smallsp_s(f)$.
Computing efficiently this decomposition is a basic task in linear 
algebra over non-archimedean fields.

In this article, we assume that we are given a routine
\texttt{Big} which takes as input $(f,s)$ and outputs (a basis of)
the subspace $\bigsp_s(f)$.
A naive implementation of the procedure \texttt{Big} is reported in 
Algorithm~\ref{algo:big}. It has cubic complexity in the dimension 
of $V$ (which will be enough for our applications) but has the
advantage of being numerically stable.

\begin{algorithm}[t]
  \caption{\texttt{Big}$(f,s)$} \label{algo:big}
  \Input{$f \in K^{ \delta \times \delta}$, $s \in \R.$}
  \Output{A basis $S$ of $\bigsp_s(f)$.}
  $\chi_f \leftarrow \texttt{charpoly}(f)$%
    \tcp*{use \cite{KV:2004} or \cite{CRV:2017}}
  Write $\chi_f = \chi_{\smallsp,s, f} \chi_{\bigsp,s,f}$%
    \tcp*{use \cite{CRV:2016}}
  $g \leftarrow \chi_{\bigsp,s, f}(f)$%
    \tcp*{use \cite{PS:1973}}
  $S \leftarrow \ker g$%
    \tcp*{use pnumerical kernel from \cite[\S 2.2.1]{KV:2020}}
\Return{$S$}  
\end{algorithm}

\subsubsection*{$\Kz$-modules and saturation.}

As before, we let $V$ be a finite dimensional $K$-vector space.
We recall basic facts about sub-$\Kz$-modules of $V$ and their
algorithmic.
If $V$ is equipped with a distinguished basis, one can represent
a finitely generated sub-$\Kz$-module of $V$ by the matrix $M$ 
whose columns are the generators of $L$.
Performing column reduction, one can always assume that $M$ is 
under Hermite normal form. With this additional assumption, it is 
uniquely determined by $L$.

Let $L_1$ and $L_2$ be sub-$\Kz$-modules of $V$, represented by the
square matrices $M_1$ and $M_2$ respectively. The sum $L = L_1
+ L_2$ is then generated by the columns of the block matrix:
$$M = \left( \begin{array}{@{\hspace{0.5ex}}c@{\hspace{0.5ex}}c@{\hspace{0.5ex}}c@{\hspace{0.5ex}}}
M_1 & | & M_2 
\end{array} \right)$$
Computing the Hermite normal form of $M$, one obtains a 
canonical matrix representating $L$.
The cost of this computation is cubic in the dimension of $V$
with a naive algorithm.

We now assume that we are given a sub-$\Kz$-module $L \subset V$ 
together with a $K$-linear endomorphism $f : V \to V$.
The \emph{saturation} of $L$ with respect to $f$ is the 
sub-$\Kz$-module of $V$ defined by:
$$\Sat_f(L) = L + f(L) + f^2(L) + \cdots + f^n(L) + \cdots$$

\begin{lem}
\label{lem:saturate}
We assume that $V = \smallsp_0(f)$. Then:
$$\Sat_f(L) = L + f(L) + f^2(L) + \cdots + f^{\dim V-1}(L).$$
In particular, if $L$ is finitely generated, then $\Sat_f(L)$ is 
also.
\end{lem}

\begin{proof}
The assumption on $f$ implies that the coefficients of $\chi_f$ 
are all in $\Kz$. From Cayley-Hamilton theorem, we deduce that
$f^\delta$ is a linear combination with coefficients in $\Kz$
of the $f^i$'s with $i < \delta$. The lemma follows.
\end{proof}

The routine \texttt{Saturate} presented in 
Algorithm~\ref{algo:saturate} computes $\Sat_f(L)$ under
the assumption that $L$ is finitely generated  and $V = 
\smallsp_0(f)$. 
Indeed, one checks by
induction that after the $k$-th iteration of the loop, one has
$g = f^{2^k}$ and:
$$S = L + f(L) + f^2(L) + \cdots + f^{2^k-1}(L).$$
Therefore, when $2^k \geq \delta$, we find $S = \Sat_f(L)$.
The complexity of Algorithm~\ref{algo:saturate} is equal to the
cost of $O(\log \delta)$ Hermite reductions. If we use the naive
algorithm for this task, we obtain an algorithm of arithmetic
complexity $\softO(\delta^3)$.

\begin{algorithm}[t]
  \caption{\texttt{Saturate}$(f, L)$} \label{algo:saturate}
  \Input{a $K$-linear map $f : V \to V$ s.t. $V = \smallsp_0(f)$,\\
         a finitely generated $\Kz$-module $L \subset V$}
  \Output{$\Sat_f(L)$}

  $S \leftarrow L$;\, $g \leftarrow f$\;
  $w \leftarrow \lceil \log_2 \dim V \rceil$\;
  \For{$k \in \llbracket 1, w \rrbracket$} {
    \label{algo:line:power_it_5}
    $S \leftarrow S + g(S)$\;
    $g \leftarrow g^2$\;		 
  }  
  \Return{ $S$ }  
\end{algorithm}

\begin{rem}
For a general $K$-linear mapping $f$, one always has:
$$\Sat_f(L) = \bigsp_0(f) + 
L + f(L) + f^2(L) + \cdots + f^{\dim V -1}(L)$$
provided that $L$ spans $V$ as a $K$-vector space. Under this
assumption, one can then combine Algorithms~\ref{algo:big}
and~\ref{algo:saturate} to compute the saturation of $L$
with respect to $f$ even when $\smallsp_0(f) \subsetneq V$.
\end{rem}

\section{Multiplication matrices}
\label{sec:mult_mat_for_Tate_alg}

Throughout this section, we fix a tuple $\r = (r_1, \ldots, r_n)$
and consider the Tate algebra $\KX[\r]$. We consider in addition
a $0$-dimensional ideal $I$ of $\KX[\r]$ and assume that we are
given a Gröbner basis $G = (g_1, \ldots, g_s)$ of $I$.

The first step in the FGLM algorithm is the computation of the
matrices of multiplication by the 
variables on the quotient $\KX[\r]/I$ (which has finite dimension by 
assumption). We recall that a $K$-basis of $\KX[\r]/I$ is given by 
the staircase $B$, which consists of all monomials $m$ with 
$m \not\in \LT(I)$.
We let $T_i$ be the matrix of the multiplication by $X_{i}$ with
respect to this basis.
Observe that the $(\mu, m)$-entry of $T_i$ has valuation at least:
$$v_{\mu, m} =  \val(mX_i) - \val(\mu)
             =  \val(m) - r_i - \val(\mu).$$
Our goal is to design an algorithm for computing the $T_i$'s.
In order to express our complexity estimates, we introduce two
important parameters. This first one is the degree of the ideal $\delta = |B| =
\dim \KX[\r]/I$. The second one, denoted by $\epsilon$, is the
size of the boundary of the staircase defined as $\bar B
\setminus B$ with
$$\bar{B} = 
\big\{\,X_{i} m \, :\, i \in \{1,\dots,n\}, m \in B\,\big\}.$$
Obviously the cardinality of $\bar B$ is at most $n\delta$; thus
$\epsilon \leq n \delta$ as well.

The theorem we are going to prove is the following (we refer to 
the beginning of Section~\ref{sec:setting} for the definition of 
the arithmetic and base complexity).

\begin{theo}
There exists an algorithm that takes as input a reduced Gröbner
basis of $G$ and outputs the multiplication matrices $T_i$ with
$(\mu,m)$-entry known at precision
$O(\pi^{\text{\rm prec} + v_{\mu,m}})$ for 
a cost of $O(\epsilon \delta^2 \text{\rm prec})$ arithmetic 
operations.

Besides, if the base field $K$ is either a Laurent series field
or $\Qp$, the above complexity can be lowered to 
$\softO(\epsilon \delta^2 \text{\rm prec})$ base operations.
\end{theo}

We will also present an algorithm accepting as input certain 
nonreduced Gröbner basis. This variant is interesting because, 
in some cases, it will eventually provide a fast algorithm for 
interreducing Gröbner basis.

\subsection{Iterative algorithm}
\label{sec:computation-ala-fglm}

Throughout this subsection, we assume that $G = (g_1, \ldots, 
g_s)$ is \emph{reduced} and $\r = (0, \ldots, 0)$. We will explain 
later on how these assumptions can be relaxed.
For simplicity, we assume in addition that the $g_i$'s are all
monic (\emph{i.e.} the coefficients of their leading terms are $1$). 
This hypothesis is of course harmless since renormalizing the 
$g_i$'s and making them monic does not affect the fact that $G$
is a Gröbner basis.

Computing the $T_i$'s amounts to computing the normal forms of $m$ modulo $I$ for all $m$ in $\bar B$.
In a classical setting, this can be done iteratively with linear algebra, by considering the monomials following the monomial order.
Indeed, for $m$ in $B$ and $i \in \{1,\dots,n\}$, if $X_{i}m \notin B$, either $X_{i}m$ is a leading monomial in $G$, or there exists $\mu \notin B$ such that $X_{i}m = X_{j}\mu$, and then $\NF(X_{i}m) = X_{j}\NF(\mu)$.
In the classical setting, the normal form of a monomial $\mu$ only involves monomials in $B$ strictly smaller than $\mu$, so $X_{j}\NF(\mu)$ only involves monomials in $\bar{B}$ strictly smaller than $X_{i}m$.
This allows to write $\NF(X_{i}m)$ as a linear combination of already computed normal forms.

In the case of Tate term orderings, similarly to what was observed for example for tropical orderings~\cite{IVY20}, the normal form of a monomial $\mu$ can involve all monomials of $B$, and computing the wanted normal forms \emph{a priori} requires solving a large nonlinear system of equations.

However, because Tate Gröbner basis are just classical Gröbner basis when they are reduced 
modulo $\pi$ (Theorem~\ref{theo:GB}), the above strategy allows to recover the value of the 
multiplication matrices modulo $\pi$.
Following the same computations again lifts the multiplication matrices to coefficients in $\Kz/\pi^{2}$, and so on and so forth.

\newcommand{\vvect}[1]{[#1]}

The algorithm formalizing this idea is described in Algorithm~\ref{algo:mult_mat_from_TateGB}.
\begin{algorithm}[t]
    \caption{\texttt{MulMat\_iter}$(G, \text{prec})$} 
    \label{algo:mult_mat_from_TateGB}
  \Input{a reduced Gröbner basis $G$ of the ideal $I \subset \KX$,\\
         an integer $\prec$ such that all elements of $G$\\
         are known at precision $\prec$}
  \Output{$T_1,\dots,T_n$, the multiplication matrices over $\KX/I$
    (w.r.t the basis $B$) modulo $\pi^{\prec}$}
  $B \leftarrow \{m \text{ monomials not divisible by any } \LT(g), g \in G
  \}$\;
  $T_{i} = (c_{i,m,m'}) \leftarrow$ zero matrices of size $\delta \times \delta$, with rows and columns indexed by $B$, for all $i \in \llbracket 1,n \rrbracket$\;
  \For{$k$ from $0$ to $\prec$\label{algo:line:beg-T}}{
  \For{$i \in \llbracket 1,n \rrbracket$, $m \in B$ in increasing order of $X_{i}m$
  }{
    \uIf{$X_{i}m \in B$}{
      $T_{i} \cdot \vvect{m} \leftarrow \vvect{X_{i}m}$\label{algo:eq-simple}\;
    }
    \uElseIf{$X_{i}m = \LT(g)$ for some $g \in G$}{
      $T_{i} \cdot \vvect{m} \leftarrow \vvect{g - \LT(g)}$\label{algo:eq-gb}\;
    }
    \Else{
      Write $m = X_{j}m'$ for some $m' \notin B$\;
      $T_{i} \cdot \vvect{m} \leftarrow T_{j}\cdot (T_{i} \cdot \vvect{m'})$\label{algo:eq-quadratic}\;
    }
  }}
  \Return{$T_{1},\dots,T_{n}$}
\end{algorithm}
For $P \in K[\mathbf{X}]$ with support contained in $B$, the notation $\vvect{P}$ represents the vector of coefficients of $P$ in the basis $B$.
With that notation, given a monomial $m \in B$ and a matrix $M$ with rows and columns indexed by $B$, $M \cdot \vvect{m}$ is the column of $M$ corresponding to $m$.

For $i \in \{1,\dots,n\}, \mu, m\in B$, we denote by $c_{i,\mu,m}$ the value at row $\mu$ and column $m$ in the multiplication matrix $T_{i}$.
The following theorem states the correctness and the complexity of the algorithm.

\begin{theo}
  \label{prop:termin-r0}
  Algorithm~\ref{algo:mult_mat_from_TateGB} is correct.
  More precisely, at the end of the $k$-th run of the loop, the matrices $(T_{i})$ are correct modulo $\pi^{k}$.
  Furthermore, each run through the loop requires $O(\delta^{2}\epsilon)$ operations in $\Kz$.
\end{theo}

The proof uses the following observation, which is the translation to the Tate setting of the structure of the normal forms of the staircase in the classical setting.

\begin{lem}
  \label{prop:structure-r0}
  If $X_{i}m \in B$, then $\val(c_{i,\mu,m}) > 0$ if $\mu \neq X_{i}m$.
  Otherwise, if $X_{i}m \leq \mu$, then $\val(c_{i,\mu,m}) > 0$.
\end{lem}
\begin{proof}
  By definition, the column indexed by $m$ in the multiplication matrix $M_{i}$ is the vector of the coordinates of the normal form $N$ of $X_{i}m$ modulo $G$, in the basis $B$.
  If $X_{i}m \in B$, then $N = X_{i}m$ and the result is clear.
  Otherwise, if $c_{i,\mu,m} \mu$ is a term of $N$, then $c_{i,\mu,m}\mu < X_{i}m$, which, by definition of the Tate term ordering, means that either $\mu < X_{i}m$ or $\val(c_{i,\mu,m}) > 0$. 
\end{proof}

\begin{proof}[Proof of the theorem]
  We prove the result by induction on $k \geq 0$, and, for each value of $k$, by induction on $X_{i}m$, $i \in \{1,\dots,n\}$, $m \in B$. 
  
  The initial case $k=0$ is empty.
  Let $k>0$, $i \in \{1,\dots,n\}$ and $m \in B$, and assume by induction that we know the coefficient $c_{j,\mu',m'}$ with precision $k$ if $X_{j}m' < X_{i}m$, and with precision $k-1$ otherwise.
  
  If $X_{i}m \in B$, then there is nothing to prove, because the coefficients are 0 or 1.
  If $X_{i}m = \LT(g)$ for some $g \in G$, there is also nothing to prove, since all the coefficients of $M_{i}\cdot \vvect{m}$ are known to precision $\prec \geq k$.
  
  In the remaining case, for each $\mu \in B$, the algorithm performs the substitution
  \begin{equation}
    \label{eq:7}
    c_{i,\mu,m} \leftarrow \sum_{\mu' \in B} c_{j,\mu,\mu'} c_{i,\mu',m'}.
  \end{equation}
  Since $X_{j}m' = m$, $m' < m$ and $X_{i}m' < X_{i}m$ and so the induction hypothesis applies. 
  Let $\mu' \in B$
  .
  Note that $X_{j}\mu' \neq X_{i}m$: otherwise, $X_{j}\mu' = X_{i}X_{j}m'$ so $\mu' = X_{i}m'$, which cannot lie in $B$.
  If $X_{j}\mu' < X_{i}m $, then both $c_{j,\mu,\mu'}$ and $c_{i,\mu',m'}$ are known up to precision $k$ (induction on $X_{i}m$), so the product is known to precision $k$.
  And if $X_{j}\mu' > X_{i}m$, then $c_{j,\mu,\mu'}$ is known up to precision $k-1$ (induction on $k$) and $c_{i,\mu',m'}$ is known up to precision $k$ (induction on $X_{i}m$) and divisible by $\pi$ (by Lemma~\ref{prop:structure-r0}), so the product is known up to precision $k$.

  For the number of operations, observe that there are less than $\epsilon$ pairs $(i,m)$ such that the alsorithm needs to perform the computation at line~\ref{algo:eq-quadratic}; each computation involves $\delta$ coefficients of the matrix, and for each of them, $\delta$ products in $\Kz$.
\end{proof}

\subsection{Nonreduced Gröbner bases}
\label{sec:nonreduced-grobner}

An interesting feature of the algorithm above is that contrary to the usual case, it has to handle monomials which are larger than the current monomial $X_{i}m$.
This removes the main reason for the requirement that the input Gröbner basis is reduced, and with slight modifications, it can handle any Gröbner basis as long as it is reduced modulo $\pi$.
Precisely, this is achieved by replacing line~\ref{algo:eq-gb} with the following. 

\bgroup
\removelatexerror
\begin{algorithm}[H]
  \renewcommand{\theAlgoLine}{\thealgocf.8\alph{AlgoLine}}
  \addtocounter{algocf}{-1}
  \renewcommand{\thealgocf}{\arabic{algocf}a}
  \caption{Update the matrices using a nonreduced basis element}
  \label{algo:nonreduced}
  \renewcommand{\thealgocf}{\arabic{algocf}}
  \For{$a\mu$ in the support of $g - \LT(g)$}{
    \uIf{$\mu < X_{i}m$}{
      $T_{i}\cdot \vvect{m}\leftarrow T_{i}\cdot\vvect{m} + a [\mu]$
    }
    \Else{
      Pick $m' \in B$ such that
      $\mu = X_{1}^{\alpha_{1}}\cdots X_{n}^{\alpha_{n}}m'$\;\label{algo:line:decom-nonred}
      $T_{i}\cdot \vvect{m} \leftarrow T_{i}\cdot \vvect{m} + aT_{1}^{\alpha_{1}}\cdots T_{n}^{\alpha_{n}} \cdot \vvect{m'}$\;
    }
  }
\end{algorithm}
\egroup

Unless the staircase is trivial, \emph{i.e.} as long as the ideal is proper, it is always possible to find a suitable $m'$ at line~\ref{algo:line:decom-nonred}, by picking the monomial $1 \in B$.
Nonetheless, to avoid computing large powers of matrices, it is more efficient to find $m'$ as large as possible.

It is still true that any monomial $\mu > X_{i}m$ appearing in the process necessarily carries a coefficient with valuation $\geq 1$, and thus the loop invariant that the coefficients are known to precision $k$ still holds.

The complexity of the computation is no longer bounded merely in terms of $\delta$, $\epsilon$ and $\prec$, but also depends on the degree of the nonreduced terms in the basis, and on the choices of the monomials $m'$.

\subsection{Recursive algorithm}
\label{sec:recursive-algorithm}

As described above, the computations can be done in increasing order of the monomials $X_{i}m$, ensuring that all the necessary coefficients are known with the necessary precision for the next step.
Another way to proceed is by dynamic programming, computing the necessary coefficients recursively if they are not known yet.

The recursive definition, using the matrices $T_{i}$ as a cache, is described in Algorithm~\ref{algo:mult_mat_rec}, and is very similar to that described in Algorithm~\ref{algo:mult_mat_from_TateGB}.

It can then be called, for all values of $i,\mu,m$ and $k=\prec$, instead of lines \ref{algo:line:beg-T}--\ref{algo:eq-quadratic} in Algorithm~\ref{algo:mult_mat_from_TateGB}.
The main difference is that the algorithm does not need to specify in which order the coefficients are computed: the recursive definition queries the missing coefficients as needed.
The decision on which precision is needed depends on the valuation of the coefficient: the idea is that if $a$ is known with precision $k$ and has valuation $v$, and $b$ is known with precision $l$ and has valuation $w$, then $a\cdot b$ is known with precision $\min(k+w,l+v)$.
Note that it also works if we only know a lower bound on the valuation, typically if all the digits we know are $0$.

The proof that the recursive algorithm terminates is the existence of such an order, as demonstrated in Theorem~\ref{prop:termin-r0}.
And the proof of complexity is also immediate: there are $\epsilon$ coefficients for which the calculation is non-trivial, and for each of them, after $\delta$ multiplications, we gain one digit of precision.
The total complexity is then $O(\epsilon \delta^2 \text{prec})$ operations in $\Kz$ as in the iterative case.

The advantage of the recursive presentation is twofold.
Firstly, it will allow in Section~\ref{sec:other-log-radii} to generalize the construction, the proof of termination, and the complexity bounds, to arbitrary log-radii.

Secondly, it offers a way to immediately improve the performance of the algorithms, on coefficient rings such as $\ZZ_{p}$ or $k(\!(T)\!)$ where fast arithmetic is available.
This works by using a lazy representation of the number, that is, a representation where each number is the data of its first digits, as well as a function allowing to compute the next digit.
Algorithm~\ref{algo:mult_mat_rec} gives us precisely such a function, and as such, the process can be viewed as a recursive definition of lazy numbers (the function definition) together with a delayed evaluation (the function call for all values).

For many coefficient rings, it is possible to do better by using the so-called relaxed, or on-line, arithmetic.
Such arithmetics are available for formal power series rings~\cite{vdH:1997} and $p$-adic numbers~\cite{BvdHL:2011,BL:2012}.
In that case, the cost of the computation of each new digit (of each variable) is polynomial in $\log(\text{prec})$ if we are counting base operations (in the sense of Section~\ref{sec:setting}).
Here, this allows us to compute the matrices with base complexity in $\softO(\epsilon \delta^{2} \prec)$.

\begin{algorithm}[t]
  \caption{\texttt{MulMat\_rec}$(G, B, i, \mu, m, k)$} 
  \label{algo:mult_mat_rec}

  \SetKwInOut{Global}{Global}
  
  \Input{$G$ as in Algo.~\ref{algo:mult_mat_from_TateGB}, $B$ the staircase of $G$,\\
         $i \in \llbracket 1,n \rrbracket,\,\mu \in B, \, m \in B, \, k \in \ZZ, k \leq \prec$}
  \Global{$(T_{i}) = (c_{i,\mu,m})_{\mu,m \in B})_{i \in \llbracket 1,n \rrbracket}$}
  \Output{$T_{i}$ is such that $c_{i,\mu,m}$ is known to precision $k$ 
  }
  \uIf{$c_{i,\mu,m}$ is known to precision $k$ in $T_{i}$}{
    \textbf{do nothing}
  }
  \uElseIf{$k \leq 0$}{
    $c_{i,\mu,m}\leftarrow O(1)$
  }
  \uElseIf{$X_{i}m \in B$}{
    $c_{i,\mu,m}\leftarrow$ 1 \textbf{if} $\mu = X_{i}m$ \textbf{else} 0}
  \uElseIf{$X_{i}m = \LT(g)$ for $g \in G$}{
    $c_{i,\mu,m}\leftarrow$ the coordinate of $\mu$ in the support of $g - \LT(g)$
  }
  \Else{
    Write $m = X_{j}m'$ for some $m' \notin B$\;
    $c \leftarrow 0$\;
    \For{$\mu' \in B$}{
      $v \leftarrow \val(c_{j,\mu,\mu'})$;
      $w \leftarrow \val(c_{i,\mu',m'})$\;
      $\texttt{MulMat\_rec}(G, B, j,\mu,\mu',k{-}w)$\;
      $\texttt{MulMat\_rec}(G, B, i,\mu',m',k{-}v)$\;
      $c \leftarrow c + c_{j,\mu,\mu'}c_{i,\mu',m'}$\; \label{line:add}
      }
    $c_{i,\mu,m} \leftarrow c + O(\pi^{k+1})$\; 
 }
\end{algorithm}

\begin{rem}
  For simplicity and for the complexity bounds, we only presented the procedure in the case where the Gröbner basis is reduced, but given that the recursive definition is equivalent to the loop presented in Algorithm~\ref{algo:mult_mat_from_TateGB}, the case where the Gröbner basis is not reduced can be dealt with in exactly the same way.
\end{rem}

\subsection{General log-radii}
\label{sec:other-log-radii}

We now consider the case of arbitrary log-radii $\mathbf{r} \in \QQ^{n}$.
We will prove that the algorithm presented above still works in that case, by using abstract changes of variables and base ring to justify the existence of a suitable execution order.
Crucially, the algorithm works without performing those transformations, and the complexity is the same.
We only need to be more careful about the handling of the precision and of the valuation.

Namely, given $\mathbf{r} \in \QQ^{n}$, we will assume that the input basis $G$ is normalized, in the sense that $0 \leq \val(\LT(g)) < 1$ for all $g \in G$.
We will further require that for each $g \in G$, and for each $t$ in the support of $G$, $t$ is known to precision $\prec + \floor{\val(t)}$, and we will ensure that we compute the matrices with similar precision by ensuring that $c_{i,\mu,m}$ is correct up to precision $k+\floor{v_{\mu,m}}
$.

Recall that $\NF(X_{i}m) = \sum_{\mu \in B} c_{i,\mu,m} \mu$ with, for all $\mu$, $c_{i,\mu,m} \mu < X_{i}m$.
So by definition of the Tate term ordering, $\val(c_{i,\mu,m}) \geq \val(X_{i}m) - \val(\mu) = v_{\mu,m}$, and the requirement on the precision is merely adjusting the number of digits we require beyond those we already know to be 0.
The only difference is that each term is initialized with the zero digits and the precision which we already know:

\vspace{0.5cm}
\bgroup
\removelatexerror
\begin{algorithm}[H]
    \addtocounter{algocf}{-1}
  \renewcommand{\thealgocf}{\arabic{algocf}a}
    \caption{Base case with non-zero log-radii}
    \label{algo:rec-log-radii}
  \renewcommand{\thealgocf}{\arabic{algocf}}
  \setcounter{AlgoLine}{2}
  \renewcommand{\theAlgoLine}{\thealgocf.\arabic{AlgoLine}}
  \uElseIf{$k \leq  \lfloor\val(X_{i}m) - \val(\mu)\rfloor$}{
    $c_{i,\mu,m}\leftarrow O(\pi^{\lfloor\val(X_{i}m) - \val(\mu)\rfloor})$\;
  }
\end{algorithm}
\egroup
\vspace{0.5cm}


\begin{theo}
  Let $\mathbf{r} \in \QQ^{n}$ be a system of log-radii.
  Algorithm~\ref{algo:mult_mat_from_TateGB}, with input a reduced Gröbner basis of an ideal in $\KzX[\mathbf{r}]$, and modified to compute the matrices using Algorithm~\ref{algo:mult_mat_rec}, computes the multiplication matrices in $O(\epsilon \delta^{2}\prec)$ multiplications in $\Kz$.
\end{theo}
\begin{proof}
  Let $\Gamma(G)$ be the dependency graph of the recurrence relation defined in Algorithm~\ref{algo:mult_mat_rec} with the modifications of Algorithm~\ref{algo:rec-log-radii}: namely, $\Gamma(G)$ is a directed graph whose vertices are tuples $(i,\mu,m,k)$, and there is a directed edge $(i,\mu,m,k) \rightarrow (j,\mu',m',l)$ if and only if the computation of $c_{i,\mu,m}$ to precision $k$ queries the coefficient $c_{j,\mu',m'}$ to precision $l$.
  Note that the vertices with no outgoing edge correspond to coefficients which are immediately known to precision $\prec$.
  The recursive computation terminates if and only if the graph is cycle-free, namely, if every path through the graph eventually reaches a vertex with no outgoing edge.

  In the case of trivial log-radii $\mathbf{r} = (0,\dots,0)$, the proof of that fact is Theorem~\ref{prop:termin-r0}.
  Assume that $\mathbf{r} \in \QQ^{n}$.
  %
  %
  Let $D$ be the common denominator of the log-radii, so that $\mathbf{r} = (r_{1}/D, \dots, r_{n}/D)$.
  Without loss of generality, we may assume that $G$ is minimal: elements of $G$ which can be removed will not take part in the computation.
  Consider the field extension $L = K[\eta]$ with $\eta^{D} = \pi$, and perform the change of variables $X_{i} \leftarrow \eta^{r_{i}}Y_{i}$.
  This change of variables transforms $G$ into a Gröbner basis $G'$ of an ideal in $\Lz\{\Y\}$.
  In this case, the algorithm terminates, so the graph $G'$ is cycle-free.
  If $G$ is minimal, so is $G'$, and by \cite[Prop.~3.10]{CVV}, the elements of this basis lie in $\Kz\{\Y\} \subset \Lz\{\Y\}$.
  In particular, all throughout the algorithm, the coefficients of the matrices are in $\Kz$.

  The graph $\Gamma(G)$ is isomorphic to a subgraph of $\Gamma(G')$, the inclusion being given by $(i,\mu,m,k + \floor{v_{m,\mu}} \rightarrow (i,\mu,m,k)$.
  Since $\Gamma(G')$ is cycle-free, so is $\Gamma(G)$ and the algorithm terminates.
  
  The bound on the number of operations can be obtained with a similar 
  argument as before, or read on the graph: the complexity is bounded
  by $2 \delta$ times the number of vertices of $\Gamma(G)$ since 
  computing each new vertex has a cost of $2 \delta$ operations in
  $K$ (the additions and multiplications on line~\ref{line:add}).
  Since $\Gamma(G)$ has at most $\epsilon\delta{\cdot}\prec$ vertices, the bound
  $O(\epsilon \delta^2 \prec)$ follows.
\end{proof}

\section{Change of log-radii and ordering}
\label{sec:FGLM_for_Tate_alg}

The next step in the FGLM algorithm consists in going in the
opposite direction: starting from multiplication matrices and a
term ordering, we aim at reconstructing the underlying Gröbner
basis.

Moreover, in our setting where we want to be able to handle in 
addition changes of log-radii, a preliminary step is needed. 
Indeed, the multiplication matrices are usually affected by a 
modification of the log-radii. For example,
the ideal generated by $2x^2- y^2$ and $y^3-x^2$
in $\Q_2[x,y]$ has staircase $\{1, y, x, y^2, xy, xy^2\}$ (for lex)
while it spans an ideal over $\Q_2 \{x,y \}$
with staircase $\{1,y\}$ (still using lex). 
We study this phenomenon in full generality in
Section~\ref{subsec:from_poly_to_tate_alg}.

A toy implementation of the algorithms of this Section
is available on \url{https://gist.github.com/TristanVaccon}.

\subsection{New multiplication matrices}
\label{subsec:from_poly_to_tate_alg}

\subsubsection*{Theoretical results}

Let $\r$ and $\u$ be two $n$-tuples
such that $r_i \geq u_i$ for all~$i$.\footnote{The results of 
this section
can be extended without difficulty to $r_i=+\infty,$ \textit{i.e.} to
$K[\X]$.} Under this assumption
the Tate algebra $\KX[\r]$ is included in $\KX[\u]$ and, given
an ideal $I$ in $\KX[\r]$, it makes sense to consider the ideal
$J = I \cdot \KX[\u]$ of $\KX[\u]$.

In what follows, we always assume that $I$ is $0$-dimensional.
The quotient $\KX[\r]$ is then, by definition, a finite dimensional
$K$-vector space; we will denote it by $V$. Similarly, we set
$W = \KX[\u]/J$. The inclusion $\KX[\r] \hookrightarrow \KX[\u]$
induces a $K$-linear mapping $\Phi : V \to W$.

In order to study $\Phi$, we use topological arguments.
We let $\Vert\cdot\Vert_\u$ be the norm on $\KX[\u]$ associated to the 
Gauss valuation $\val_\u$ and equip $\KX[\u]$ with the topology
associated to this norm.

\begin{lem}
\label{lem:closure}
The ideal $J$ is the closure of $I$ in $\KX[\u]$.
\end{lem}

\begin{proof}
The polynomial ring $K[\mathbf X]$ is dense in $\KX[\u]$ for the
norm $\Vert\cdot\Vert_u$. Therefore, $\KX[\r]$ is dense as well,
implying that $I$ is dense in $J$.
The fact that $J$ is closed follows from~\cite[Chap.~2, Cor.~8]{bosch}.
\end{proof}

The norm $\Vert\cdot\Vert_\u$ induces by restriction a norm on
$\KX[\r]$ (which is, of course, different from the standard
norm $\Vert\cdot\Vert_\r$ on this space) and a mapping
$\Vert \cdot \Vert_V : V \to \R^+$ defined by:
$$\Vert x \Vert_V = \inf_{\hat x} \Vert \hat x \Vert_\u$$
where the infinum runs over all $\hat x \in \KX[\r]$ lifting~$x$.
In general, $\Vert \cdot \Vert_V$ is not a norm but only a 
semi-norm, meaning that there might exist elements $x \in V$ for
which $\Vert x \Vert_V = 0$. By definition, the \emph{kernel} of 
$\Vert \cdot \Vert_V$ is the set of such elements; we denote it
by $N$. It is easily seen that $N$ is a sub-$K$-vector space of
$V$.

\begin{prop}
\label{prop:Phi}
The map $\Phi$ is surjective and its kernel is~$N$.
\end{prop}

\begin{proof}
We notice that $\KX[\u]$ is the completion of $\KX[\r]$ for
the norm $\Vert \cdot \Vert_u$. Combining this observation with
Lemma~\ref{lem:closure}, we deduce that $W$
appears as the completion of $V$ with respect to the semi-norm
$\Vert \cdot \Vert_V$, which is also the completion of $V/N$.
But, since $V/N$ is finite dimensional, it is already complete.
As a conclusion, $W \simeq V/N$ and the proposition is proved.
\end{proof}

We now assume that we are given the multiplication matrices
$T_1, \ldots, T_n$ over $V$. We want to relate them to $W$, or
equivalently to~$N$. This is the content of the following
proposition.

\begin{prop}
\label{prop:N_in_the_bigTi}
With the above notations, we have:
\[ N = \sum_{i=1}^n \bigsp_{u_i}(T_i) \]
(where we recall that the notation $\bigsp_{u_i}$ was defined
in \S \ref{ssec:linalg}).
\end{prop}

\begin{proof}
First of all, we observe that, up to replacing $K$ by $K[\pi^{1/D}]$ for 
a well-chosen integer $D$, we can assume without loss of generality that 
$\u$ is in $\ZZ^n$. Replacing $T_i$ by $\pi^{-u_i} T_i$ and $\r$ by
$\r{-}\u$, we may further suppose that $\u = (0, \ldots, 0)$.

Let $i \in \{1, \ldots, n\}$ and let $x \in \bigsp_0(T_i)$.
By definition, $x$ is killed by $\chi_{\bigsp,0,T_i}(T_i)$.
In other words, if $f \in \KX[\r]$ is a lifting of $x$, the
product $\chi_{\bigsp,0,T_i}(X_i) \cdot f$ lies in $I$.
Now, we claim that $\chi_{\bigsp, 0, T_i}(X_i)$ is invertible in
$\KX$ because it is a product of factors of the form
$a^{-1} (1-aX_i)$ with $\val(a) > 0$. Consequently, $f$ must
be an element of $J$. By Proposition~\ref{prop:Phi}, we derive 
$x \in N$, which proves the inclusion $\bigsp_0(T_i) \subset N$.
Since this holds for any $i$, the $\supset$ part of the
Proposition is proved.

Set $N' = \sum_{i=1}^n \bigsp_0(T_i)$ and $W' = V/N'$.
From what we have done so far, we deduce that the semi-norm $\Vert 
\cdot \Vert_V$ on $V$ induces a semi-norm on $W'$. The proposition
will follow if we can prove that $\Vert \cdot \Vert_V$ is indeed a
norm (\emph{i.e.} with trivial kernel) on $W'$. In order to do so,
we consider the unit ball of $W'$, namely:
$$D' = \big\{\, x \in W' \,\, \text{s.t.} \,\,
\Vert x \Vert_V \leq 1\,\big\}.$$
We want to prove that $D'$ does not contain any $K$-line.
For this, we remember that the unit ball of $\KX$ is exactly the
$\Kz$-module generated by the monomials $\X^{\mathbf i}$ for 
$\mathbf i$ varying in $\NN^n$. Therefore, $D'$ is the smallest 
$\Kz$-module stable under the $T_i$'s and containing the image
of $1 \in \KX$ in $W'$. Keeping in mind in addition that the
$T_i$'s commute pairwise, we get
$D' = \Sat_{T_1} \Sat_{T_2} \cdots \Sat_{T_n} (L_0)$
where $L_0$ is the sub-$\Kz$-module of $W'$ generated by the image
of~$1$.

Besides, on $W'$, all the eigenvalues of all the $T_i$'s 
have nonnegative valuation since we have quotiented out all the
$\bigsp_0(T_i)$'s. Consequently Lemma~\ref{lem:saturate} applies 
and shows that $D'$ is finitely generated. In particular, it
contains no $K$-line, as wanted.
\end{proof}

\subsubsection*{Explicit computations}

\begin{algorithm}[t]
  \caption{\texttt{NewMulMat}$(T_1, \ldots, T_n, v)$}
  \label{algo:new_mult_mat}
  \Input{$T_1,\dots,T_n$ the multiplication matrices over $V$,\\
         $v$ the image of $1 \in \KX[\r]$ in $V$,
         $\u \in \ZZ^n$}
  \Output{$U_1,\dots,U_n$
          the multiplication matrices 
          over the unit ball of $W$,\\
          $w$ the image of $1 \in \KX[\r]$ in $W$}
  $N \leftarrow \{0\}$\;

  \For{$i \in \llbracket 1,n \rrbracket$}
  {
    $N \leftarrow N + \texttt{Big}_0(T_i)$\;
  }
  $W \leftarrow V/N$\;
  $L \leftarrow v \Kz$\;
  \For{$i \in \llbracket 1,n \rrbracket$}
  {
    $L \leftarrow \texttt{Saturate}(T_i,\, L)$\;
  }
  \textbf{return} $T_{1|L},\dots, T_{n|L}, v \mod N$\;
\end{algorithm}

It is straightforward to turn the previous theoretical analysis
into an actual algorithm that computes the space $W \simeq V/N$ 
and the multiplication matrices acting on it. In fact, for later use, 
it will not be enough to express these matrices in any $K$-basis of $W$, 
but we shall really need a $\Kz$-basis of the unit ball of~$W$ (for the 
norm $\Vert \cdot \Vert_V$ introduced before).

When $\u = (0, \ldots, 0)$, Algorithm~\ref{algo:new_mult_mat} does
the job.
In the description of this algorithm, we have implicitely assumed 
that all $K$-vector spaces and $\Kz$-modules are equipped with 
distinguished bases, and consequently used the same notation for a matrix 
and the endomophism it represents. All operations on $\Kz$-modules 
can be handled using Hermite normal forms as recalled in \S 
\ref{ssec:linalg}; similarly, operations on $K$-vectors spaces can be 
done using Smith normal forms, which permits to keep better numerical 
stability.

If $\delta$ denotes the dimension of $V = \KX[\r]/I$ (which is also the 
size of the matrices $T_i$'s), Algorithm~\ref{algo:new_mult_mat} 
requires at most $\softO(n \delta^3)$ operations in the base field~$K$.

\subsection{Reconstruction of the Gröbner basis}
\label{subsec:FGLM_for_Tate_alg}

The final step in the FGLM algorithm is the computation of a Gröbner 
basis from the datum of the multiplication matrices.

\subsubsection*{Trivial log-radii}

We first address the case where $\u = (0, \ldots, 0)$, which is
covered by Algorithm~\ref{algo:tate_fglm} (page \pageref{algo:tate_fglm}).
This algorithm uses a routine \texttt{FGLMField} which takes
as input a set of $n$ multiplication matrices over a field (together
with the vector representing the monomial~$1$) and a term ordering
and returns the corresponding Gröbner basis.
A description of such an algorithm performing this task can be found
in many places in the litterature, for example in the original article
by Faugère \emph{et al.}~\cite{Faugere:1993}.

\begin{algorithm}[t]
  \caption{\texttt{GB}$(U_1, \ldots, U_n, w, \leq)$} \label{algo:tate_fglm}
  \Input{$U_1,\dots,U_n$ the multiplication matrices over the 
         unit ball of $W$,\\
         $w$ the image of $1 \in \KX[\r]$ in $W$,\\
         $\leq$ a monomial ordering}
  \Output{A Gröbner basis $G$ of the ideal $J \subset \KX[\u]$}

  $\bar G \leftarrow \texttt{FGLMField}(U_1 \mod \pi, \ldots, U_n \mod \pi, w \mod \pi, \leq)$\;

  $B \leftarrow \{m \text{ monomials not divisible by any } \LM(g), g \in \bar G\}$

  $M = (M_{\star,\mu})_{\star \in \mathcal B, \mu \in B} \leftarrow \text{zero matrix}$\;
  \tcp{$\mathcal B$ denotes the distinguished basis of $W$ we are working with}
  $M_{\star, 1} \leftarrow w$\;

  \For{$\mu \in B \backslash \{1\}$ by increasing order for $\leq$}
  {
    write $\mu = X_i \mu'$ with $i \in \{1, \ldots, n\}$, $\mu' \in B$\;
    $M_{\star,\mu} \leftarrow U_i \cdot M_{\star,\mu'}$%
     \tcp*{product matrix-vector}
  }

  $N = (N_{\star,m})_{\star \in \mathcal B, m \in \LM(\bar G)} \leftarrow \text{zero matrix}$\;

  \For{$m \in \LM(\bar G)$}
  {
    write $m = X_i \mu$ with $i \in \{1, \ldots, n\}$, $\mu \in B$\;
    $N_{\star,m} \leftarrow U_i \cdot M_{\star,\mu}$%
     \tcp*{product matrix-vector}
  }

  $Q\leftarrow M^{-1}N$\;\label{algo:line:fglm-3}
  $G \leftarrow \left( m - \sum_{\mu \in B} Q_{\mu,m} \mu \right)_{m \in \LM(\bar G)}$\;
  \Return{$G$}
\end{algorithm}

\begin{prop}
  Algorithm \ref{algo:tate_fglm} is correct 
  and runs in $O(n \delta^3)$ arithmetic operations
  where $\delta$ denotes the dimension of $W$.
\end{prop}

\begin{proof}
We recall that we assume $\u = (0, \ldots, 0)$. 
Let $J = I{\cdot}\KX$ be as in
Section~\ref{subsec:from_poly_to_tate_alg}. We recall that $W =
\KX/J$ by definition. Let $D$ denote the unit ball of $W$. From 
the facts that the unit ball of $\KX$ is $\KzX$ and the norm on $W 
\simeq \KX/J$ is the quotient norm, we deduce that $D \simeq \KzX/
J^\circ$ with $J^\circ = J \cap \KzX$.
The reductions modulo $\pi$ of the $U_i$'s are then the multiplication
matrices on the quotient $\bar J = J^\circ/\pi J^\circ$.
The call to \texttt{FGLMField} then returns a Gröbner basis of the
ideal $\bar J$.
From Theorem~\ref{theo:GB}.(2), we derive that the leading terms
of a Gröbner basis of $J$ are formed by the monomials in $\LM(\bar G)$.
It follows from this that $B$ is the staircase of the ideal $J$.
In particular, its cardinality is the dimension of $W$, showing that
the matrix $M$ is a square matrix.
After the loops, the columns of $M$ (resp. of $N$) contain the
coordinates of the $\mu$'s (resp. the $m$'s) in the distinguished 
basis $\mathcal B$ for $\mu$ varying in~$B$ (resp. for $m$ varying
in~$\LT(G)$).
The matrix $Q = M^{-1} N$ then contains the expression of the $m$'s
in terms of linear combination of the $\mu$'s. This shows the
correctness of the algorithm.

The fact that the complexity is in $O(n \delta^3)$ arithmetic 
operations is easily checked.
\end{proof}

\subsubsection*{General log-radii}

We now consider the general case where $\u = (u_1, \ldots, u_n)
\in \QQ^n$.
We take $D \in \ZZ_{>0}$ to be a common denominator
of the coordinates of $\u$ (consequently $D \cdot \u \in \ZZ^n$).
We define the field extension $L=K[\eta]$ such that $\eta^D=\pi$
and perform the change of variables $\tilde X_i = \eta^{D u_i} X_i$.
The Tate algebra $L \otimes_K \KX[\u]$ becomes isomorphic to
$L\{\mathbf{\tilde X}\}$ and we can then apply all what precedes
with $L \otimes_K \KX[\u]$.

Inside $L\{\mathbf{\tilde X}\} \simeq L \otimes_K \KX[\u]$ sits 
the subset $\eta^\ZZ \KX[\u]$ consisting of series of the form $\eta^v 
f$ with $v \in \ZZ$ and $f \in \KX[\u]$.
Let $I_L$, $J_L$, $V_L$ and $W_L$ denote the spaces deduces by
$I$, $J$, $V$ and $W$ respectively by extending scalars from $K$
to $L$. Inside them, we can similarly define $\eta^\ZZ I$, 
$\eta^\ZZ J$, $\eta^\ZZ V$ and $\eta^\ZZ W$.
We claim that then Algorithms~\ref{algo:new_mult_mat}
and~\ref{algo:tate_fglm} can be adapted so that they only have
to manipulate vectors lying in these subsets. Indeed:

\begin{itemize}[leftmargin=3ex]
\item
the $\bigsp_0(\tilde T_i)$'s can be computed without passing
to~$L$ because they are equal to the $\bigsp_{u_i}(T_i)$'s which are
defined over $K$;
\item
similarly the quotient $V/N$ is defined over $K$ and then
does not create any trouble;
\item
one checks that the Hermite normal form of a matrix whose
column vectors are in $\eta^\ZZ V$, remains of this form;
\item
the column vectors of the matrices $M$ and $N$ of
Algorithm~\eqref{algo:tate_fglm} all come from monomials and so
have the required shape.
\end{itemize}

\noindent
Proceeding this way, we avoid the time penalty due to scalar extension 
from $K$ to $L$ and keep a complexity of $O(n \delta^3)$ arithmetic 
operations.
At the end of the day, the output of Algorithm~\ref{algo:tate_fglm}
is then a Gröbner basis $G$ of $J_L$ consisting of series in
$\eta^\ZZ \KX[\u]$. Still staying in the same subset, we can 
normalize these series so that they all have Gauss valuation~$0$.
In this case, $G$ is not only a Gröbner basis of
$J_L$ but also a Gröbner basis on the ideal $J_L^\circ = L
\otimes_K J^\circ = J_L \cap \KzX[\u]$.
From~\cite[Proposition 3.10]{CVV}, we deduce that $G \cap \KX[\u]$
remains a Gröbner basis of $J^\circ$ and hence of $J$.

\subsubsection*{Conclusion}

Combining Algorithms~\ref{algo:mult_mat_from_TateGB} (or 
\ref{algo:mult_mat_rec}), \ref{algo:new_mult_mat} and~\ref{algo:tate_fglm} 
and the above discussion for covering the case of arbitrary log-radii, we finally end 
up with a complete FGLM algorithm as announced in Theorem~\ref{theo:intro}
in the introduction.
Plugging Algorithm~\ref{algo:nonreduced} into the machine, we notice 
that our algorithm can also accept Gröbner bases which are nonreduced 
as soon as they are reduced modulo the maximal ideal. However, in this 
case, the complexity may grow up rapidly, depending on the shape of the 
input Gröbner basis.



%
%


\bibliographystyle{plain}

\end{document}
